\documentclass[conference]{IEEEtran}
\IEEEoverridecommandlockouts
\usepackage{cite}
\usepackage{amsmath,amssymb,amsfonts}
\usepackage{graphicx}
\usepackage{textcomp}
\usepackage{xcolor}
\def\BibTeX{{\rm B\kern-.05em{\sc i\kern-.025em b}\kern-.08em
    T\kern-.1667em\lower.7ex\hbox{E}\kern-.125emX}}
\usepackage{amsthm}
\usepackage{graphicx}
\usepackage[colorlinks=true, allcolors=blue]{hyperref}
\usepackage{subcaption}

\usepackage{algorithm}
\usepackage{algpseudocode}

\makeatletter
\newenvironment{interface}[1][!htb]{%
    \renewcommand{\ALG@name}{Interface}
   \begin{algorithm}[#1]%
  }{\end{algorithm}}
\makeatother

\usepackage{framed}
\newtheorem{lemma}{Lemma}
\newtheorem{definition}{Definition}

\usepackage{listings}

\algnewcommand\algorithmicuponmsg{\textbf{Upon receiving message}}
\algnewcommand\Uponmsg{\item[\algorithmicuponmsg]}
\algnewcommand\algorithmicuponevent{\textbf{Upon event}}
\algnewcommand\Uponevent{\item[\algorithmicuponevent]}
\algnewcommand\algorithmicupon{\textbf{Upon}}
\algnewcommand\Upon{\item[\algorithmicupon]}
\algnewcommand\algorithmicsendmsg{\textbf{send message}}
\algnewcommand\sendmsg{\State\algorithmicsendmsg}
\algrenewcommand\algorithmicfunction{\textbf{function}}
\algnewcommand\function{\item[\algorithmicfunction]}
\algnewcommand\algorithmicstate{\textbf{state}}
\algnewcommand\state{\item[\algorithmicstate]}
\algnewcommand\algorithmicemitevent{\textbf{Emit event }}
\algnewcommand\Emitevent{\State\algorithmicemitevent}
\algnewcommand\algorithmicbroadcast{\textbf{broadcast}}
\algnewcommand\broadcast{\State\algorithmicbroadcast}

\algnewcommand\algorithmicstruct{\textbf{struct}}
\algnewcommand\struct{\item[\algorithmicstruct]}

\begin{document}

\title{Defending against the nothing-at-stake problem in multi-threaded blockchains\\}

\author{\IEEEauthorblockN{1\textsuperscript{st} Léonard Lys}
\IEEEauthorblockA{
\textit{Massa Labs}\\
Paris, France\\
ll@massa.net\\
0000-0002-7080-9409}
\and
\IEEEauthorblockN{2\textsuperscript{nd} Sébastien Forestier}
\IEEEauthorblockA{
\textit{Massa Labs}\\
Paris, France\\
sf@massa.net\\
0000-0002-2702-8869
}
\and
\IEEEauthorblockN{3\textsuperscript{rd} Damir Vodenicarevic}
\IEEEauthorblockA{
\textit{Massa Labs}\\
Paris, France \\
dv@massa.net \\
0000-0001-9843-9043
}
\and
\IEEEauthorblockN{4\textsuperscript{th} Adrien Laversanne-Finot}
\IEEEauthorblockA{
\textit{Massa Labs}\\
Paris, France\\
alf@massa.net}
}

\maketitle

\begin{abstract}
In blockchain systems, the scarcity of a resource is used as a Sybil protection mechanism. In Proof-of-Work blockchains, that resource is computing power. In the event of a fork, the scarcity of this resource theoretically prevents miners from producing blocks on both branches of a fork. In Proof-of-Stake blockchains, because that resource is a token stake, the computational cost of creating a block is negligible. In the event of a fork, and if no specific measures have been taken, rational block producers should extend both branches of the fork. In blockchains with sequential block production, a punishment mechanism known as slashing is often cited as a protection against the nothing-at-stake problem. However, in the context of a blockchain with parallel block production, it seems that slashing is not sufficient against the numerous divergence opportunities. In this paper, we propose a novel protection against the nothing-at-stake problem that takes the most out of BFT and Nakamoto-based consensus. By combining those approaches, we wish to scale up blockchains by allowing parallel block production without reconciliation.
\end{abstract}


\section{Introduction}
Currently, the most popular blockchain systems have a very limited throughput in terms of transactions per second. Bitcoin and Ethereum, the two biggest chains in terms of market capitalization achieve less than 20 transactions per second. Considering that payment systems such as Visa achieve 1700 transactions per second, one could argue that blockchains are facing a scalability problem.

In sequential blockchains such as Bitcoin or Ethereum, block producers wait for the previous block to be delivered before producing their own block. Because of this sequential block production, the throughput of sequential blockchains is limited by the block propagation times. In order to solve this scalability issue, many projects such as \cite{yu2020ohie} or \cite{forestier2018blockclique} have adopted a parallel architecture. The idea is to have several parallel instances of a sequential blockchain and a reconciliation mechanism that guarantees consistency among the parallel chains. This architecture allows for parallel block production and thus solves the scalability issues of sequential blockchains.

In Proof-of-Stake blockchains, the computational cost of creating a block is negligible, which is different from Proof-of-Work blockchains. When a node is selected to produce a block, it can thus create multiple blocks with different hashes for the same slot, while only one can be taken into account by other nodes in that slot. Similarly, when a fork occurs, the block producer can technically continue to produce blocks on all branches of the fork. This problem is commonly known in the literature as the nothing-at-stake problem. Moreover, if no particular protection is taken, an attacker can bloat the network by creating thousands or millions of valid blocks in all of its slots. This attack could make it harder for honest nodes to get to a consensus on the executed blocks, even more, if the attacker has a large proportion of the block production power, say $\beta = 1/3$. We refer to this type of flooding attack as a Double- or multi-staking attack.

Sequential Proof-of-Stake with Nakamoto-style consensus blockchains such as Tezos~\cite{goodman2014tezos} seems to have no particular protection against this attack other than an incentive mechanism with a double-staking punishment and a small Proof-of-Work challenge.

In sequential blockchains, simply waiting for the next block producer to choose among the alternative chains might be sufficient to avoid consequences on liveness. However, it may not be for the parallel chains, due to the numerous parallel opportunities for the attacker to create incompatible blocks. For instance, if we consider a parallel blockchain with $T$ threads and given a hypothesis of a block propagation time of $\delta$, for each block creation opportunity for the attacker, the next $T$ blocks will have different parents. As a result, the number of incompatible blocks grows exponentially for each block creation opportunity. Agreeing on the ``best'' chain is inherently a slow process as block creators need to be aware of all possible parents. As a result, if the attacker has a large proportion of the stake, the number of incompatible blocks will grow faster than the number of final blocks due to honest nodes agreeing on the ``best'' chain.

The multi-staking attack has two consequences:
\begin{itemize}
    \item massively increases the network usage if all blocks, headers, and operations are propagated,
    \item massively increases CPU usage if blocks, headers, and operations are verified (signature verification, etc.).
\end{itemize}
Because this behavior is easily detectable, a denunciation and punishment system can be implemented. We could for example, "slash" the stake of the validator that signs several versions of a block. While this punishment mechanism would prevent rational actors from multi-staking, it is not the case for the modeled attacker that seeks to disturb the protocol functioning for their benefit or even at their own cost. Moreover, the punishment strategy does nothing to fix the damage done. If an attacker broadcasts too many blocks for the other nodes to process, and even though the attacker's roll will eventually be destroyed, the damage has already been done as the honest cannot decide on this block.
Thus it is necessary for the protocol's safety to find a strategy that allows honest participants to decide, even in the presence of such attackers.

In order to defend against the multi-staking attack in parallel blockchains, we propose a protocol that takes the most out of BFT and Nakamoto consensus. The idea is to decouple a BFT-like block production system from a Nakamoto-like finalization system. 

To do so, we propose that for each block slot, a sub-committee of validators called "endorsers" is drawn. Those endorsers will vote on the block proposed by the block producer. A block will be considered valid, if and only if it receives more than a threshold number of endorsements, for example, $2/3$ of the total number of endorsements. Simple computation show (see Appendix~\ref{app:super-majority-size}) that it is possible to find a reasonable set of parameters for which it is extremely unlikely that an attacker is able to meet the threshold of endorsements by himself. For instance, with $E = 96$ endorsements and a threshold of $64$ endorsements, the probability that an attacker having $1/3$ of the stake meets the threshold by himself is around $p = 5.3e-12$ (given one block every 500ms).

\section{Model and definitions}
We consider a distributed system consisting of an arbitrary finite set of processes $\Pi$. Processes exchange messages through peer-to-peer and bi-directional communication channels. The goal of the protocol is to maintain an eventually consistent ledger among the network.

\subsection{System model}
\paragraph{Cryptography}
Processes use digital signatures to identify each other across the network; each process possesses a pair of public and private key that they use to sign or verify each message they send or deliver. We assume an ideal cryptography model where signatures cannot be forged and collisions do not exist. 

\paragraph{Ledger model}
The ledger $\mathcal{L}$ is a key value store that associates each public key with a data store. Processes store a local replica of the ledger. It is assumed that the ledger is updated in rounds. The state of the ledger $\mathcal{L}$ is the result of the sequential execution of operations called transactions. Transactions are batched into limited-size chunks called blocks. Thus each block $b_i$ represents a totally ordered list of operations that modify the state of the ledger. The ledger $\mathcal{L}$ supports two operations; $EXECUTE(\mathcal{L}, b_i)$ which atomically execute the operations contained in $b_i$ to $\mathcal{L}$ and $READ(\mathcal{L}, i)$ which returns the state of the ledger $\mathcal{L}$ after the execution of block $b_i$. The state of the ledger $\mathcal{L}$ after the execution of block $b_i$ is noted $\mathcal{L}_i$.

\paragraph{Arrival model}
The size of the set $\Pi_p \subset \Pi$ that participate in the protocol is not a priori-known, however similarly with \cite{amoussou2018correctness} we consider a subset $V \subset \Pi$ of processes called validators. The size $n$ of the set $V$ is fixed and known before each round. Processes can be promoted to $V$ on a merit parameter. This could be modeled as the stake in a proof-of-stake blockchain system. 

\paragraph{Failure model}
Processes can be correct or faulty. A correct process follows the protocol while a faulty one exhibit Byzantine behavior. There is no bound on the number of Byzantine processes of $\Pi$. However, we assume that at most $f=n/3$ processes of $V$ can exhibit Byzantine failures. We assume here that $f$ is expressed as a portion of the so-called merit parameter. In the case of a proof-of-stake system, $f$ would be expressed in a portion of the total stake. 

\paragraph{Communication network}
It has been shown in  \cite{mills1995improved, elson2002fine} that algorithms such as NTP \cite{mills1985ntp} can keep clock skew within tens of microseconds in a wide area with a negligible cost. Thus we assume a partially synchronous model where clocks are loosely synchronized. To broadcast a message to the network, it is assumed that processes have access to a $broadcast$ primitive. When a correct process sends a message by invoking the $broadcast$ primitive, all correct processes deliver the message within a time $\delta = t_0$.

\paragraph{Execution model}
Time is divided into epochs which are subdivided into slots. At each slot $s$, a process $p_v \in V$ is selected to produce a block, and a subset $\Pi_e \subset V$ is selected to produce endorsements for this block. We assume that processes have access to a deterministic selection function $committee(s)$ that for a given target slot $s$ computes $p_v$ and $\Pi_e$. The selection function $committee$ is based on decided values such that the output of $committee$ for a given slot $s$ should output the same result for every correct process. On receiving a block $b_i$ signed by $p_v$, correct processes append $b_i$ to a locally maintained block graph. For each new block $b_i$ appended to the block graph, correct processes apply a Nakamoto-like consensus rule to decide which block of the graph is to be finalized and thus executed to the ledger $\mathcal{L}$.

\subsection{Blockchain model}
\subsubsection{Threads and block slots}
We define $T$ threads numbered from $\tau=0$ to $\tau=T-1$, each containing consecutive slots that can host blocks, regularly spaced by $t_0$ seconds. The $i$-th block slot ($i\in \mathbb{N}$) in thread $\tau$ is denoted by $s^\tau_i$ and occurs at $i\cdot t_0 + \tau t_0/T$ seconds. The $\tau t_0/T$ time shift ensures that block slots, and therefore network usage, are uniformly spread in time.

\subsubsection{Multithreaded Block DAG structure}
\label{multithreaded}
Blocks are identified by their cryptographic hash: $b_{h}^\tau$ refers to the block with hash $h$ in thread $\tau$.
We define one genesis block with no parents, denoted  $b_0^\tau$, in the first block slot of each thread.
Each non-genesis block references the hashes of $T$ parent blocks, one from each thread.

We define the block graph $G$ as the graph where each vertex represents a block and each directed edge represents a child-parent relation.
A directed edge between two blocks indicates a child-parent relation, i.e., a directed edge from $b^{\tau_1}$ to $b^{\tau_2}$ indicates that $b^{\tau_2}$ is the parent of $b^{\tau_1}$ in thread $\tau_2$.

As a child block includes its parent's hashes, and the hashing procedure of the child block takes into account those hashes, it is impossible to build a cycle in the graph $G$, unless the security of the hashing function is compromised.
$G$ is therefore a directed acyclic graph of parallel blocks (multithreaded block DAG).

\begin{definition}
Let $G$ denote a block DAG structure with all the following properties:
\begin{itemize}
\item one genesis block is present in each of $T$ threads,
\item non-genesis blocks in thread $\tau$ reference one block of each thread as parents, and have a timestamp strictly higher than any of their parents,
\item to ensure the consistency of block references, any ancestor $b_{h_1}^{\tau_1}$ of a block $b_{h_2}^{\tau_2}$ must be the parent of $b_{h_2}^{\tau_2}$ in thread $\tau_1)$ or one of its ancestors,
\item blocks have a size lower than $S_B$ bits.
\end{itemize}
\end{definition}

\subsubsection{Transaction Sharding}
In our protocol, transactions are sharded: they are deterministically divided into groups to be processed in parallel threads.
For instance, if there are $T=32$ threads, the first $5$ bits of an address define the thread in which transactions originating from this address can be included.

\begin{definition}
Given the sets of possible addresses $\mathcal{A}$ and transactions $\mathcal{T}$, a \emph{sharding} function $\mathcal{S}$ uniformly assigns any address $a\in\mathcal{A}$ to a particular thread $\mathcal{S}(a)=\tau\in[0,T-1]$, and any transaction $\mathrm{tx}\in\mathcal{T}$ to the thread assigned to the transaction's emitter address. The transaction $\mathrm{tx}$ can only be included in a block of thread $\mathcal{S}(\mathrm{tx})$, and can only reduce the balance of addresses assigned to this thread.
\end{definition}

Transaction sharding ensures that transactions in a block are compatible with transactions in blocks from other threads as they can't spend the same coins. 
We stress that this restriction only applies to spending, and transactions can send coins to any address, regardless of the thread it is assigned to. Transactions in a thread are regularly taken into account in blocks of other threads through parent links so that no further cross-shard communication is required.

\subsection{Blockchain properties}
The goal of the protocol is to implement an eventually consistent ledger among a distributed network of processes. We use a definition of eventual consistency similar to the literature standard \cite{saito2005optimistic} in a probabilistic decision model. Our definition of chain head and the committed prefix is also derived from the concepts of candidate chain and finalized chain from the formalization of the Ethereum proof-of-stake in \cite{ethereum2022pavloff}. 

\begin{definition}
    Chain head: The graph of blocks that have not been finalized yet. Considering the view of the block graph of an honest process $p_i$, $p_i$’s associated chain head is noted $G_i^{head}$. 
\end{definition}
Note: The view of the head of the honest process $p_i$ can differ from the view of other honest processes.

\begin{definition}
    Committed prefix: The committed prefix is the graph constituted of all the finalized blocks. The committed prefix is always a prefix of any chain head $G_i^{head}$.
\end{definition}

\begin{definition}
    Final block: A block $b_i$ is final for a process $p_i$ when it cannot be revoked with v.h.p. At each round, processes apply a finality condition to decide which blocks are final. The state changes implied by a final block to the ledger are applied by invoking  $EXECUTE(\mathcal{L}, b_i)$.
\end{definition}
Note: Because the finality condition depends on the depth of a block, all predecessors of a final block are final.

\section{Super-majority among sub-committee multi-staking protection}
\subsection{High-level protocol description}
In order to choose a block $b_{p}$ as a thread parent, some child block $b_{c}$ of will have to include more than the threshold number of endorsements endorsing $b_{p}$. If some block receives more than the threshold number of endorsements, the next block producer will be able to form what we call a super-majority certificate. This certificate, which is a list of endorsements plus some metadata on the endorsed block, will be included by the block producer in the child block. If some block $b_p$ does not reach the threshold number of endorsements, it cannot be selected as a parent and will be discarded. If an attacker intents a multi-staking attack, most of the proposed blocks will not reach the threshold number of endorsements and thus, can safely be ignored by the other nodes.
In the canonical case, the scenario is as follows:
\begin{enumerate}
    \item For each slot $s_i$, one block producer and $E$ endorsers are pseudo-randomly drawn according to their proportion of the stake.
    \item The block producer creates a block $b_i$ that includes a super-majority certificate $cert_{i-1}$  endorsing the parent block of $b_i$, and broadcasts it.
    \item On receiving block $b_i$, if it is valid, endorsers sign the content of the block and broadcast the signature. This signature serves as an endorsement for this block.
    \item Node runners propagate the block and the endorsements to their peers.
    \item As a block producer for the following slot $s_{i+1}$, I try to collect more than the threshold number of endorsements to form a super-majority certificate $cert_i$ endorsing $b_i$. If I achieve to do so, I select $b_i$ as a parent for my block $b_{i+1}$ within which I include the super-majority certificate $cert_i$ I just formed.
\end{enumerate}
In this protocol, it is still easy for an attacker to create multiple versions of a block. However, as those blocks are not going to be taken into account by block producers as long as they do not meet the endorsements threshold, nodes can simply ignore them. In the unlikely case where two (or more) blocks were produced and one of them ends up meeting the endorsements threshold, a node that does not possess the endorsed block can simply query peers to get it. As this block meets the endorsement threshold, it means that, statistically, a large portion of honest nodes have the block in their storage and can share it with other peers. Those procedures are detailed in the algorithms below.

\subsection{Protocol specification}
In this section, we are going to specify the protocol implemented by the validator's clients. The specification is articulated as follows:
Algorithm \ref{alg:blockclique} implements the fork choice rule of the protocol. It computes the heaviest, fully-connected sub-graph upon which block producers should build. 
Algorithm \ref{alg:scheduler} is the Scheduler. It is responsible for making the Proof-of-Stake pseudo-random draws that determine for each slot, the set of endorsers as well as the block producer. When the current process public key is drawn, the Scheduler launches a corresponding instance of the Endorser \ref{alg:endorser} or Block producer \ref{alg:block-producer} Algorithm. 
Algorithm \ref{alg:message-handler} is responsible for handling incoming messages. It implements the propagation rules that make the protocol safe against flooding attacks. It checks the validity of the messages and decides whether an incoming block or endorsement should be processed or ignored. 
Interface \ref{alg:utilities-interface} and \ref{alg:data-structures} specify some data structures and utilities employed in the other Algorithms. Full description of the interface \ref{alg:utilities-interface} is available in appendix \ref{alg:utilities-full}. 

\subsubsection{Data structures and communication}
\paragraph{Data structures}
In the Interface \ref{alg:data-structures}, we present the different data structures used in the rest of the specification. 
\begin{interface}[!htb]
\caption{Data structures}
\label{alg:data-structures}
\footnotesize
\begin{algorithmic}[1]
\struct{slot}
    \State $thread$ : Slot's thread 
    \State $period$ : Slot's period
\\
\struct{parent}
    \State $hash$ : Hash of the block
    \State $slot$ : Slot of the block
\\
\struct{Endorsement}
    \State $slot$ : Slot in which the endorsement can be included
    \State $index$ : Endorsement index
    \State $endorsedBlock$ : Hash of the endorsed block
\\
\struct{Certificate}
    \State $slot$ : Slot where the endorsements of the certificate were produced
    \State $endorsements$ : Vector of endorsements
    \State $endorsedBlock$ : Hash of the endorsed block
\\
\struct{Block}
    \State $slot$ : Slot in which the block can be included
    \State $parents$ : Vector of parents indexed by thread
    \State $certificates$ : Vector of super-majority certificates included in this block
    \State $operations$ : Vector of operations
\end{algorithmic}
\end{interface}

\paragraph{Broadcast primitive}
To broadcast a message to the network, processes invoke the $broadcast$ primitive. The $broadcast$ primitive expects two primary fields; a tag that defines the type of message and a payload. The payload itself is subdivided into several fields that depend on the type of message. Thus a broadcast primitive invocation is a form $boradcast(<TAG>, msg)$, where $<TAG>$ defines the type of message among \{BLOCK, ENDORSEMENT, CERTIFICATE, REQUEST-BLOCK\}. 

\subsubsection{Fork choice rule: the blockclique}
\label{subsection:fork-choice}
In sequential blockchains with Nakamoto consensus, when two blocks reference the same parent, we say that they are incompatible. Incompatible blocks generate alternative branches of a fork, and among those branches, only one version can be appended to the committed prefix. In the context of our multi-threaded block DAG, incompatible blocks generate sub-graphs called cliques. Whereas in sequential blockchains with Nakamoto consensus, block producers build upon the longest chain, in our protocol, block producers build upon the best clique called the Blockclique.
Algorithm \ref{alg:blockclique} implements the fork choice rule of the protocol. It is a weighted Nakamoto longest-chain rule extended to our parallel block graph architecture. It computes the heaviest fully connected sub-graph of compatible blocks.  

\paragraph{Compatibility Graph}
In sequential blockchains with Nakamoto-style consensus, if some blocks $b_i$ and $b*_{i}$ reference the same parent, any block $b_{i+1}$ referencing $b_i$ as a parent will be incompatible with $b*_{i}$ and all of its descendants. In the context of our multi-threaded block DAG, the same rule applies, but incompatibilities are inherited to all descendants across threads. In parallel blockchains, blocks are produced in parallel, i.e., some block producer for slot $s_i^{\tau_2}$ does not have to wait to deliver block $b_i^{\tau_1}$ to produce his block $b_i^{\tau_1}$. While this property is desirable for the protocol throughput, it would generate incompatibilities according to sequential blockchain compatibility rules. Thus, the compatibility rules must be extended to allow for parallel block production. 

In order to represent compatibilities, we define the compatibility graph $G_C$ as the graph that links all compatible blocks of $G$ with an undirected edge. In this graph $G_C$, if two vertexes are not connected with an undirected edge, they are incompatible. From this compatibility graph $G_C$, we will be able to compute the heaviest fully connected graph, i.e., the Blockclique, among which block producers should build upon. Because super-majority certificates also account for the weight of a sub-graph, we also append super-majority certificates to the compatibility graph $G_C$. 

The compatibility graph $G_C$ is constructed in two steps, first, all blocks are appended, and then all super-majority certificates.

To construct the graph $G_C$, we sequentially append all blocks $b_i$ of $G$ in topological order. For each new block $b_i$ appended to $G_C$, we sequentially compare $b_i$ to all vertexes $b_j$ of $G_C$, with $j$ in topological order, and link $b_i$ and $b_j$ with an undirected edge if they verify one of the following condition:
\begin{itemize}
    \item Ancestor compatibility
        \begin{equation}
        \label{eq:ancestorCompatibility}
        \begin{cases}
            \text{All parents of $b_i$ are connected with $b_j$ in $G_C$}, \\
            \land directedPath(G, b_i, b_j) 
        \end{cases}       
    \end{equation}
Where the predicate $directedPath(G, b_i, b_j)$ indicates that there is a directed path from $b_i$ to $b_j$ in $G$. This rule reflects the fact that blocks that are ancestors and descendants of each other should be compatible if they reference mutually compatible parents.

    \item Parallel compatibility
        \begin{equation}
        \label{eq:parallelCompatibility}
            \begin{cases}
                \text{All parents of $b_i$ are connected with $b_j$ in $G_C$}, \\
                \land |t(b_i) - t(b_j)| < t_0, \\
                \land threadParent(v_i) \neq  threadParent(v_j)
            \end{cases}       
        \end{equation}
Where $t(b_i)$ is the time of the slot of the block of $b_i$. This rule reflects the fact that if two blocks of different threads were produced in a time frame smaller than the propagation time $t_0$, then they should be compatible given that they have compatible parents and that they do not have the same thread parent.
\end{itemize}

Once all blocks have been appended to the compatibility graph $G_C$, super-majority certificates are appended in topological order. We make a distinction between two types of super-majority certificates; the ones that have been included in a block, and the speculative ones, which can be formed but are not yet included in a block. Certificates that are included in a block inherit the compatibilities of the block they were included in, while speculative certificates inherit the compatibilities of the block that they endorse. 

\paragraph{Fitness}
In the protocol, each block and each super-majority certificate represents one unit of weight called fitness. The fitness of a clique is defined as the number of vertexes it includes. This value reflects the number of times a super-majority certificate or a block has been produced in a clique. See Algorithm \ref{alg:blockclique}.

\paragraph{Best Clique of Compatible Blocks}
Let \texttt{cliques}$(G_C)$ be the set of maximal cliques of compatible vertexes: the set of subsets $C$ of $G_C$ so that every two distinct vertexes of $C$ are adjacent in $G_C$ and the addition of any other vertex from $G_C$ to $C$ breaks this property.
In the remainder of the paper, the term ``clique" refers to a maximal clique of compatible vertexes.
The protocol consensus rule states that the best clique that nodes should extend, is called the \emph{blockclique} and is the clique of maximum total fitness. If two cliques have the same total fitness, the clique with the smallest arbitrary-precision sum of the hashes of the blocks it contains is preferred.

\begin{algorithm*}[!htb]
\caption{Blockclique computation}
\label{alg:blockclique}
\footnotesize
\begin{algorithmic}[1]
\state
\State $G$ : Block graph
\State $G_C$ : compatibility graph
\State $Certificates$ : 2D Vector of super-majority certificates $Certificates[slot][blockHash] \rightarrow Certificate$ \Comment{2D vector because there might be multi-endorsing}

\\
\function{$blockclique()$} \Comment{Return the blockclique}
\State $blockclique \gets$ heaviest fully connected sub-graph $\in G_C$
\State \Return $blockclique$
\\
\function{$maxCliques()$} \Comment{Return the set of max cliques}
\State $cliques \gets$ vector of maximal cliques $\in G_C$
\State \Return $cliques$
\\
\function{$appendToG(b_i)$}
\State $G.addVertex(b_i)$ \Comment{Create a vertex in $G$}
\For{$parent \in b_i.parents$} \Comment{For each parent}
    \State $G.addDirectedEdge(b_i, parent)$ \Comment{Add directed edge in $G$ between block and parent}
\EndFor
\\
\function{$appendToG_C(b_i)$}
\State $G_C.addVertex(b_i)$\Comment{Create a vertex}
\For{$b_j \in G_C$}\Comment{For each block of $G_C$}
    \If{$parentsMutuallyCompatible(b_i, b_j)$}\Comment{Verify that all parents of $b_i$ are compatible with $b_j$}
        \If{$directedPath(G, b_i, b_j)$}\Comment{If they are ancestor compatible}
            \State $G_C.addEdge(b_i, b_j)$\Comment{Create an undirected path in $G_C$}
        \EndIf
        \If{$(|t(b_i) - t(b_j)| < t_0) \land (threadParent(v_i) \neq  threadParent(v_j))$} \Comment{If they are parallel compatible blocks}
            \State $G_C.addEdge(b_i, b_j)$ \Comment{Create an undirected path in $G_C$}
        \EndIf
    \EndIf
\EndFor
\For{$block \in G_C$} \Comment{For each block of $G_C$ in topological order}
    \For{$cert \in b.certificates$} \Comment{For each cert. in the block}
        \For{$neighbor \in G_C.getNeighbors(block)$} \Comment{For each compatible block}
            \State $G_C.addEdge(cert, neighbor)$ \Comment{Create an undirected path in $G_C$}
        \EndFor
    \EndFor
\EndFor
\For{$cert_i \in Certificates$}\Comment{For each certificate}
    \If{$\nexists cert_j \in G_C$ where $cert_j.endorsedBlock = cert_i.endorsedBlock \land cert_j.slot = cert_i.slot$}\Comment{If it is not included in a block nor duplicate}
        \State $G_C.addVertex(cert_i)$\Comment{Add the cert. to $G_C$}
        \For{$neighbor \in G_C.getNeighbors(cert_i.endorsedBlock$}\Comment{For each neighbor of endorsed block}
            \State $G_C.addEdge(cert_i, neighbor)$\Comment{Create an undirected path in $G_C$}
        \EndFor
    \EndIf
\EndFor
\\
\function{$parentsMutuallyCompatible(b_i, b_j)$}
\For{$parent \in b_i.parents$}
    \State \textbf{require} $G_C.getEdge(parent, b_j)$
\EndFor
\State \Return $true$

\end{algorithmic}
\end{algorithm*}

\subsubsection{Primitives and utilities}
Processes have access to a set of validity primitives. Those primitives are invoked by the node runners when processing incoming messages. There is one validity primitive per data structure; block, endorsement, certificate, and denunciation. For all of those data structures, the validity primitive checks if the message is correctly signed and structured. Then a specific logic is implemented for each type of message. Those primitives are presented in Interface \ref{alg:validity-primitives}.

\paragraph{Endorsement and certificate validity primitive}
Besides signature checks, the endorsement validity primitive verifies that the public key that signed the endorsement was indeed drawn as the endorser for the specific slot and PoS-index. For each included endorsement, the certificate validity primitive invokes the endorsement validity primitive. It then proceeds to verify that all included endorsements endorse the same block and that the threshold number of endorsements is reached. Finally, it ensures that the certificate metadata (slot and endorsed block) is consistent with the included endorsements.

\paragraph{Block validity primitive}
The block validity primitive implements the logic specific to block validity. Besides the structure and signatures checks, it verifies that the message's source was indeed drawn as a block producer. It verifies that the block was not produced too early and then proceeds to check the included certificates. For each included certificate, it invokes the certificate validity primitive and verifies that all the included certificates endorse the thread-parent of the block. The block validity primitive ensures that, among the included certificates, at least one was produced during the same period as the thread parent. 
In order to maintain consistency in the compatibility graph, the block validity primitive also checks that the block can be safely appended to the graph; first, it verifies that all referenced parents are already in the graph. Indeed, because the compatibility graph must be computed in topological order, we cannot process a block if one of his ancestors is missing. Thus, if a block references an unknown parent, the block won't be processed until the parent is delivered. The block validity primitive also verifies that all referenced parents are mutually compatible. In other words, it checks that the intersection of all parents belongs to some clique. If at least two of its parent are not in the same clique, it is unnecessary to process the block as it cannot be finalized. Finally, the block validity primitive verifies that the block's grandparents are at least older than its parent. This rule ensures that blocks acknowledge blocks created in other threads.
\\
To put it in simple terms, a block is valid if, (1) the number of included certificates is $\geq 1$, (2) all included certificates are valid, (3) all included certificates endorse the thread parent, (4) among those certificates, at least one was produced during the same period as the thread-parent, (5) the block's parents are older than the block itself, (6) all parents have been delivered (7) all parents connected in $G_C$, and (8), grandparents are older or equally old as the block's parents. 

\begin{algorithm}[!htb]
\caption{Validity primitives}
\label{alg:validity-primitives}
\footnotesize
\begin{algorithmic}[1]
\state
\State $G$ : Block graph
\State $G_C$ : compatibility graph
\State $time$ : Number of milliseconds since the genesis block of thread 0
\State $pubKey$ : This process public key
\State $threshold$ : Threshold number of endorsements to form a super-majority certificate

\\
\function{$isValid(Block : b)$} \Comment{Checks if block is valid}
    \State \textbf{require} $b.slot \geq timeToSlot(time)$ \Comment{Verify that block is not produced too early}
    \State \textbf{require} Valid signature $\land$ Valid PoS draw \Comment{Check sig and committee}
    \State \textbf{require} $parentsMutuallyCompatible(b)$ \Comment{Verify that all parents are mutually compatible}
    \State \textbf{require} $grandParentOlderThanParent(b)$ \Comment{Verify that all parents are more or equally recent than grandparent}
    \State \textbf{require} $certificatesEndorseRightParent(b)$ \Comment{Verify that certificates endorse thread parent and that parent reach the threshold in the right slot}
    \State \Return $true$
\\
\function{$isValid(Endorsement : e)$} \Comment{Checks if endorsement is valid}
    \State \textbf{require} Valid signature $\land$ Valid PoS draw \Comment{Check sig and committee}
    \State \Return $true$
\\
\function{$isValid(Certificate : cert)$}
    \State \textbf{require} $\lvert cert.endorsements \rvert \geq threshold$ \Comment{Verify that cert reached the threshold}
    \State $visitedIndex \gets \emptyset$ \Comment{Index list to check for endorsement reentrancy}
    \For{$e \in cert.endorsements$} \Comment{For each endorsement of the certificate}
        \State \textbf{require} $isValid(e)$ \Comment{Verify that endorsement is valid}
        \State \textbf{require} $e.endorsedBlock = cert.endorsedBlock$ \Comment{{Verify that endorsement endorses the right block}}
        \State \textbf{require} $e.slot = cert.slot$ \Comment{Verify that endorsement was produced at the right slot}
        \State \textbf{require} $e.index \notin visitedIndex$ \Comment{Check not duplicate}
        \State $visitedIndex.append(e.index)$ \Comment{Append index}
    \EndFor
    \State \Return $true$
\end{algorithmic}
\end{algorithm}

\paragraph{General utilities}
Besides the validity primitives, node runners have access to the utilities presented in Interface \ref{alg:utilities-interface}. This interface implements the PoS draws for both the block producer and endorsers set and provides some serializers. A full description of the interface is available in Appendix \ref{alg:utilities-full}.

\begin{interface}[!htb]
\caption{Utilities}
\label{alg:utilities-interface}
\footnotesize
\begin{algorithmic}[1]
\function{$timeToSlot(time)$} \Comment{Transforms time into slot}
    \State \Return $slot$ \Comment{Returns slot}
\\
\function{$buildCert(endorsements)$} \Comment{Tries to build a super-majority certificate out of an endorsement list}
    \State \Return $cert$ \Comment{Return certificate if possible, false otherwise}
\\
\function{$committeeBlockProducer(Slot: s)$} \Comment{PoS draw block producer}
    \State \Return $blockProducer$ \Comment{Returns the block producer public key}
\\
\function{$committeeEndorsers(Slot: s)$} \Comment{PoS draw endorsers given slot and ledger state}
    \State \Return $endorsers$ \Comment{Returns a vector of objects $(pubKey, index)$}
\\
\function{$parentsOlderThanBlock(Block: b)$} \Comment{Verifies that a block's parents are older than the block}
\State \Return $true$

\\
\function{$parentsMutuallyCompatible(Block: b)$} \Comment{Checks if a block's parents are all mutually compatible}
\State \Return $true$  \Comment{Returns true if compatible, false otherwise}

\\
\function{$grandParentOlderThanParent(Block: b)$}\Comment{Verifies that a block\'s grandparent are older than it\'s parent}
\State \Return $true$ \Comment{Returns true if older or equally old, false otherwise}

\\
\function{$certificatesEndorseRightParent(Block: b)$}
\State \Return $true$ \Comment{Returns true if endorses right parent, false otherwise}
\end{algorithmic}
\end{interface}

\subsubsection{Scheduler}
The Scheduler, Algorithm \ref{alg:scheduler} orchestrates the PoS draws to instantiate the Block producer and Endorser instances. At the beginning of every period, i.e., every $t_0/T$, it computes the PoS draws to determine if the current process is block producer, with $t_0$ being the time between two slots of the same thread and $T$ the total number of threads. Thus $t_0/T$ is the time between two consecutive slots of different threads. The scheduler also computes the PoS draws to determine the set of endorsers. It does so at the same frequency but with a $t_0/2$ delay. This delay is the modeled maximum message propagation time. It ensures that the block to be endorsed had enough time to propagate in the network. Note that, if the block is received sooner, the endorser can endorse before that $t_0/2$ delay. The block message will trigger the creation of an endorser instance. 
\begin{algorithm}[!htb]
\caption{Scheduler}
\label{alg:scheduler}
\footnotesize
\begin{algorithmic}[1]
\state
\State $time$ : Number of milliseconds since the genesis block of thread 0
\State $pubKey$ : This process public key
\State $t_0$ : Constant seconds per period (Curr. 16s)
\State $T$ : Constant amount of threads (Curr. 32)
\\
\function $schedulerBlockProducer(Slot : s)$
    \State $blockProducer \gets committeeBlockProducer(s)$ \Comment{PoS draw}
    \If{$pubKey = blockProducer$} \Comment{If current process drawn as block producer}
        \State \textbf{initialize} instance block producer($s$)
    \EndIf
\\
\function $schedulerEndorser(Slot : s)$
    \State $endorsers \gets committeeEndorsers(s)$ \Comment{PoS draw}
    \For{$e \in endorsers$}
        \If{$pubKey = e.pubKey$} \Comment{If current process drawn as endorser}
            \State \textbf{initialize} instance endorser($s, e.index$)
        \EndIf
    \EndFor
\\
\function $\bf main()$
    \State $setInterval(t_0/T, schedulerBlockProducer(timeToSlot(time))$ \Comment{Calls block scheduler every $t_0/T$, i.e. at the beginning of each period}
    \\
    \State $setInterval((t_0/2)+(t_0/T), schedulerEndorser(timeToSlot(time))$ \Comment{Calls endorser scheduler every $t_0/T$ with a $t_0/2$ delay, i.e. at the middle of each period}
\end{algorithmic}
\end{algorithm}

\subsubsection{First-in messages propagation rule}
Algorithm \ref{alg:message-handler} describes the propagation rules of the protocol.
As explained earlier, one aspect of the multi-staking attack is that, because the cost of producing a block is negligible, an attacker could flood the network with block proposals. If every proposed block were to be processed and propagated by every node, this would increase network and CPU usage and thus, potentially lead to a Denial of Service type of attack. This threat also applies to endorsement messages. To prevent such attacks, the protocol will enforce a first-in propagation rule. The idea behind this propagation rule is that, in the canonical case, an honest validator will only broadcast a single version of a message per slot where he was PoS drawn. Any subsequent message for the same slot is sufficient to detect Byzantine behavior. In the non-canonical case where an attacker broadcasts a second version of a block or an endorsement, an honest validator will store the message as proof of Byzantine behavior. It will then propagate the message in order for other nodes to acknowledge the behavior. Subsequent messages concerning the same slot will be ignored. If an honest validator has detected a Byzantine behavior, and if this validator is a block producer in the next slots, he can include in the operation of the block a denunciation containing proof of the Byzantine behavior. If the proof is valid, and if it was not yet included in a block, the Byzantine validator will see its stake slashed. Half of it will be burned and the other half given to the validator that submitted the denunciation as a reward.

Thus the message propagation rule can be loosely defined as follows: For a given block slot $s$ and corresponding Proof-of-Stake drawn block producer, node runners will only store and propagate the first block message they receive, given that this message is correctly signed and structured and given that the content of the block is valid according to the validity primitive.
Similarly, for a given endorser set and corresponding endorsement slot, node runner will only store and propagate the first endorsement message they receive given that this message is correctly signed and structured.
Any subsequent message concerning the same slot will be considered Byzantine behavior and will potentially trigger a denunciation from the recipient. The signed messages will be held as proof of Byzantine behavior. 

While this propagation rule would make the system safe against denial of service types of threat, it does not offer opportunities for honest validators to recover from a fork. Let us assume that the network gets partitioned into $P$ and $P'$, where validators of $P$ have received block $b$ first while members of $P'$ have received $b'$ first. Because of the first-in propagation rule, block $b$ resp. $b'$ have been ignored by node runners of $P$ resp $P'$. In order for the partitions to eventually converge, there must be an exception to the communication rules that allow for some multi-staked blocks to be processed and propagated without opening the door for flooding attacks.

We propose a locally maintained request list that, from the point of view of each node runner, references the hash of the requested blocks. If an incoming block is referenced in the request list, it will be processed, even though its corresponding slot is already full. Node runners will append a block hash to this list in several cases;
\begin{enumerate}
    \item If they receive enough endorsements to form a super-majority certificate for a block that is not in the block store.
    \item If they receive a valid block whose parent is not in the block store. In this case, they will request the said parent block.
    \item If they receive a valid super-majority certificate that endorses a block that is not in the block store.
\end{enumerate}
Behind the choice of those exceptions is still the same idea that the protocol should be safe against flooding attacks. It should be difficult for an attacker to forge multiple versions of a message that would trigger events (1), (2), and (3). Because all of those messages require the threshold number of endorsements to be reached, simple computation presented in \ref{app:super-majority-size} shows that it is very unlikely for an attacker to be able to forge those messages by himself. To forge such messages, the attacker has to obtain some endorsements from honest participants, which limits the number of messages he can forge and thus protects the protocol from flooding attacks. 

Note: For the sake of simplicity and readability, we did not specify how processes should answer block requests. However, the behavior is quite straightforward and can be described as follow, on receiving a block request, and if I have the specified block, I answer with a block message containing the block.

In the event of a successful multi-staking attack or during a fork, if several partitions of the network have a divergent view of the chain, the exceptions described here will allow honest node runners to discover the view of the other partitions. From there, they will be able to evaluate each branch of the fork according to the fork choice rule specified in Section \ref{subsection:fork-choice}.

The message propagation rules are further detailed in Algorithm~\ref{alg:message-handler}.

\begin{algorithm*}[!htb]
\caption{Message handler}
\label{alg:message-handler}
\footnotesize
\begin{algorithmic}[1]
\state
\State $G$: Block graph
\State $G_C$: Compatibility graph
\State $Endorsements$ : 2D vector of endorsements $Endorsements[slot][index]
\rightarrow Endorsement$ \Comment{2D vector because for each block slot there are E endorsement slots}
\State $Certificates$ : 2D Vector of super-majority certificates $Certificates[slot][blockHash] \rightarrow Certificate$ \Comment{2D vector because there might be multi-endorsing}
\State $blockDenunciations$ : Vector of Byzantine blocks
\State $endorsementDenunciations$ : 2D Vector of Byzantine endorsements
\State $pubKey$ : Current process public key
\State $requestList$ : Vector of requested blocks 

\\
\Uponmsg [BLOCK, $b$] \textbf{do}
\If{$G[b.slot] = \emptyset \parallel hash(b) \in requestList$} \Comment{Check that slot is empty i.e., no multistaking, or requested}
    \State \textbf{require} $isValid(b)$ \Comment{Check block validity. See utilities Algo. \ref{alg:utilities-full}}
    \State $missingParent \gets false$ \Comment{Boolean to check for missing parents}
    \For{$parent \in b.parents$} \Comment{For each parent}
        \If{$parent \notin G$}\Comment{If I don\'t have the parent}
            \State $requestList.append(parent)$ \Comment{Append block to req list}
            \broadcast[REQUEST-BLOCK, $parent$] \Comment{Request parent block}
            \State $missingParent \gets true$
        \EndIf
    \EndFor
    \State \textbf{require} $\neg missingParent$ \Comment{Check that no parent was missing}
    \broadcast[Block, $b$] \Comment{No problem detected, broadcast block to peers}
    \State $G.appendToG(b)$ \Comment{Store block}
    \State $G_C.appendToG_C(b)$ \Comment{Update compatibility graph}
    \State $finalizer()$ \Comment{Call the finalizer. See Algo. \ref{alg:finalizer}}
    \State $endorsers \gets committeeBlockProducer(b.slot)$ \Comment{PoS draw}
    \For{$key \in endorsers$ where $key = pubKey$}\Comment{If current process drawn}
        \State \textbf{initialize} instance endorser($b.slot, key.index$) \Comment{Don't wait for $t_0/2$}
    \EndFor
\ElsIf{$blockDenunciations[b.slot] = \emptyset$} \Comment{Only store one proof per slot}
    \State $blockDenunciations.append(b)$
\EndIf

\\
\Uponmsg [ENDORSEMENT, $e$] \textbf{do}
\If{$Endorsements[e.slot][e.index] = \emptyset$} \Comment{Check that no multi-endorsing}
    \State \textbf{require} $isValid(e)$ \Comment{Check validity. See utilities algo \ref{alg:utilities-full}}
    \State $Endorsements[e.slot][e.index] \gets e$ \Comment{Store endorsement}
    \broadcast[ENDORSEMENT, $e$] \Comment{Broadcast endorsement to peers}
    \If{$Certificates[e.slot][e.endorsedBlock] = \emptyset$} \Comment{If I don't have a certificate for this block}
        \State $cert \gets buildCert(Endorsements[e.slot])$ \Comment{Try to build a cert. with the floating endorsements}
        \If{$cert$} \Comment{If cert. can be formed}
            \State $Certificates[e.slot][e.endorsedBlock] \gets cert$ \Comment{Store cert.}
            \If{$e.endorsedBlock \in G$} \Comment{If I have the block}
                \broadcast[CERTIFICATE, $cert$] \Comment{Broadcast to peers}
            \Else \Comment{If I don\'t have the block}
                \State $requestList.append(e.endorsedBlock)$ \Comment{Add to request list}
                \broadcast[REQUEST-BLOCK, $e.endorsedBlock$] \Comment{Request block}
            \EndIf
        \EndIf
    \EndIf
\ElsIf{$endorsementDenunciations[e.slot][e.index] = \emptyset$} \Comment{Only store one proof per slot}
        \State $endorsementDenunciations.append(e)$
        \broadcast[ENDORSEMENT, $e$] \Comment{Broadcast it as proof of Byzantine behavior}
\EndIf
\\
\Uponmsg [CERTIFICATE, $cert$] \textbf{do}
\If{$Certificates[cert.slot][cert.endorsedBlock] = \emptyset$} \Comment{If I don't have a cert. for this block}
    \State \textbf{require} $isValid(cert)$ \Comment{Check cert. validity}
    \State $Certificates[cert.slot][cert.endorsedBlock] \gets cert$ \Comment{Store cert.}
    \broadcast[CERTIFICATE, $cert$] \Comment{Broadcast cert. to peers }
    \If{$cert.endorsedBlock \notin G$} \Comment{If I don't have the corresponding block}
        \State $requestList.append(cert.endorsedBlock)$ \Comment{Add block to request list}
        \broadcast[REQUEST-BLOCK, $cert.endorsedBlock$] \Comment{Request block to peers}
    \EndIf
\EndIf
\end{algorithmic}
\end{algorithm*}

\subsubsection{How endorser chose what block to endorse}
Algorithm \ref{alg:endorser} describes the behavior that endorsers should adopt to choose which block to endorse. To choose the block they endorse, endorsers compute the blockclique based on the current state of their storage. They endorse the last block of the blockclique of the thread within which they were selected to produce an endorsement. The block that will be endorsed is not necessarily the most recent block of the thread. Indeed, if for some reason, the last proposed block for a given thread is not in the blockclique, then the endorser can endorse an older block. There are actually no constraints on what block endorser can endorse, besides the thread constraint. Indeed, if endorsers were limited in their choice, for example, if they had to endorse the most recent block they have for a given thread, then any rational endorser would produce an endorsement for this block, even though this block is not the best parent in their view. Endorsements indicate to block producers what is the best block to build upon. Thus endorser can freely choose what block to endorse.

\subsubsection{How block producers create blocks}
Algorithm \ref{alg:block-producer} describes how block producers create their blocks. To select the parents, block producers compute the blockclique out of their block store, and select each last block of each thread in the blockclique for which they are able to form a super-majority certificate. 
Note: these majority certificates can be formed out of endorsements that are not yet included in a block. Block producers try to form and include as many super-majority certificates as they can. 
In the canonical case where some block $b_i$ was produced during period $i$, and this block received more than the threshold number of endorsements to form a super-majority certificate $cert_i$, then the block producer for the period $i+1$ will only include $cert_i$ in his block $b_{i+1}$ and select $b_i$ as its thread parent. 
But if for some reason no block was proposed for slot $s_i$ or if the proposed block is not in the blockclique, block producers can choose an older block as a parent, and if possible, include more than one super-majority certificate in their block.

\subsubsection{Incentives: Rewards and Penalties}
\label{sec:incentives}
In order to motivate processes to participate in the consensus with the behavior specified by the protocol, an incentive model provides rewards for appropriate behavior and penalties for deviations from the protocol.
The creation of blocks is rewarded by newly created coins.
The block reward also contains a constant amount per included endorsement, shared between the block producer, the endorsers, and the producer of the endorsed block, which motivates block creation and transmission of created blocks as well as endorsement creation and inclusion. 
The inclusion of transactions in blocks is rewarded by the fees from all included transactions. 
In order to prevent block producers from creating or endorsing multiple incompatible blocks in the same slot, we assume that the incentive model gives penalties to addresses involved in this misbehavior.
For instance, in Tezos those penalties are implemented by requiring block and endorsement producers to deposit an amount of coins that they are not allowed to withdraw for a given time~\cite{tezosdoc}.
Any node can produce a denunciation operation containing the proof that the same address has produced or endorsed multiple blocks at the same block slot, see Algorithm \ref{alg:message-handler}.
A denunciation included in a block causes a coin penalty to the offender, taken from its deposit, half of which is destroyed, and half of which is transferred to the block creator.

\subsubsection{Block finalization}
As explained earlier, in an attempt to maintain block production even in the presence of Byzantine processes, the protocol decouples block production from block finalization. While block production relies on a BFT-style algorithm, block finalization remains probabilistic in a Nakamoto-style consensus. Every time a block is appended to the graph $G$, we update the compatibility graph $G_C$. After each update of $G_C$, the total weight (e.e., the fitness) of each maximal clique is subject to change. Recall that \texttt{cliques}$(G_C)$ is the set of maximal cliques of compatible vertexes. 

\begin{figure}
\begin{definition}
Given a current compatibility graph $G_C^\mathrm{head}$ and a new block $b_{h}^{\tau}$, the Nakamoto-style consensus rule outputs a set of final and stale blocks to be removed from $G_C^\mathrm{head}$, and the blockclique to be considered.
\end{definition}
\end{figure}

A block $b_{h}^{\tau}$ is considered stale if it is included only in cliques of $G_C^\mathrm{head}$ that have a total fitness lower than the fitness of the blockclique minus a constant $\Delta_f$. 
Any new block with stale parents is considered stale. 
A block $b_{h}^{\tau}$ is considered final if it is included in all maximal cliques of $G_C^\mathrm{head}$ and included in at least one clique where the descendants of $b_{h}^{\tau}$ accumulate a total fitness greater than $\Delta_f$.

As in other blockchains with Nakamoto consensus, and because finality remains probabilistic, we let the choice of the parameter $\Delta_f$ to the appreciation of the client, depending on the desired security levels. Algorithm \ref{alg:finalizer} implements the block finalization process.

\begin{algorithm}[ht]
\caption{Finalizer}
\label{alg:finalizer}
\footnotesize
\begin{algorithmic}[1]
\state
\State $G_C$ : Compatibility graph
\State $\mathcal{L}$ : The ledger
\\
\function{\textbf{main}$()$}
\State $blockclique \gets blockclique()$ \Comment{Compute the blockclique}
\State $maxCliques \gets maxCliques()$ \Comment{Compute all max cliques}
\State $inAllMaxCliques \gets true$ \Comment{Boolean to check if in all max cliques}
\State $maxFitness \gets 0$ \Comment{Variable to store max fitness}
\For{$block \in G_C$}
    \For{$clique \in maxClique$}
        \If{$block \notin clique$}
            \State $inAllMaxCliques \gets false$
        \EndIf
        \If{$fitness(block, clique) > maxFitness$} \Comment{If the cumulative fitness of the descendants is greater than max fitness}
            \State $maxFitness \gets fitness(block, clique)$ \Comment{Update max fitness}
        \EndIf
    \EndFor
    \If{$(\neg inAllMaxCliques) \land (maxFitness < fitness(blockclique)-\Delta_f)$}\Comment{If block satisfies the stale conditions}
        \State $G.remove(block)$ \Comment{Remove stale block}
        \State $G_C.remove(block)$ 
    \EndIf
    \If{$(maxFitness > \Delta_f) \land inAllMaxCliques$}\Comment{If block satisfies the finality condition}
        \State $EXECUTE(\mathcal{L}, block)$ \Comment{Execute block}
        \State $G.remove(block)$ \Comment{Remove executed block}
        \State $G_C.remove(block)$
    \EndIf
\EndFor
\end{algorithmic}
\end{algorithm}

\section{Conclusion}
Against the current limitation of Blockchain systems in terms of throughput, we propose a protocol that allows for parallel block production. While parallelization is a straightforward solution when it comes to scaling up computing processes, it often generates increased complexity and opens up for new threats. In our case, the threat came in the form of a generalized nothing-at-stake problem that we called a multi-staking attack. We briefly analyzed the attack and proposed an innovative solution combining BFT-style block production and Nakamoto-style finality. This is to our knowledge the first protocol resistant against nothing-at-stake attacks in multi-threaded blockchains. Future works would include proposing a probability model of finalization time.

\bibliographystyle{alpha}
\bibliography{references}

\newpage

\begin{appendices}
\section{Protocol analysis}
\paragraph{Protocol liveness}
We define liveness as the ability of the chain to ever grow. Similarly, we define a liveness attack as the event where an attacker tries to prevent the chain from growing. In order to express the protocol's liveness, we define the liveness parameter as the ratio between the number of blocks an attacker can prevent from being added to the chain and the number of blocks that will be added to the chain by correct processes. 

Obviously, when the attacker is selected as a block producer or when he is drawn more than $E-threshold$ times as an endorser for a given block, he can simply withhold its block or endorsements to prevent the chain from growing. However, the average frequency of those events is known and thus allows us to model the worst-case liveness parameter of the protocol.

Recall that the modeled attacker controls $1/3$ of the stake, knows the topology of the network, is distributed geographically, and always behaves in the optimal way to perform a given attack. It is assumed that the attacker can act in or against its own financial interest for the sole purpose of disrupting the network. 

Given $E = 100$ and $threshold = 67$, such an attacker will be selected as block producer with probability $p_1 = 1/3$ and will have more than $E-threshold$ endorsements for a given block with a probability $p_2 = 1/2$. The combined probability of those events gives us the liveness parameter of the protocol:
\begin{equation}
        liveness = 1 - (p_1 \cup p_2) = 1 - p_1 + p_2 - (p_1 \cap p_2) = 1/3
\end{equation}
In other words, in the worst modeled case, the attacker can prevent two blocks out of three from being added to the chain. 

\paragraph{Forks}
When selected as a block producer, the attacker can fork the chain by selecting as thread parent, a block that is not the most recent from the point of view of the honest. Let us assume that from the point of view of a correct process, block $b_i^\tau$ is the most recent block of thread $\tau$ and that $b_i^\tau$ has received more than the threshold number of endorsements. Let us assume that the attacker is selected as a block producer at slot $s_{i+1}^\tau$. When producing its block $b_{i+1}^\tau$, the attacker can, instead of selecting $b_i^\tau$ as thread parent, select some older block $b_j^\tau$. This fork attack is a specific case of the liveness attack we discussed previously. But there is an additional threat with fork attacks that we are going to develop here. 

If block $b_i^\tau$ has already received a super-majority certificate when the attacker launches its attack, it is possible that, in other threads, honest block producers have already selected $b_i^\tau$ as a non-thread parent for their block. Indeed, the non-thread parent selection rule state that the block producer should select as a non-thread parent, the most recent block of each thread that received a super-majority certificate. As a result, if the attacked branch (the one containing $b_{i+1}^\tau$) is finally adopted by the honest, rather than the branch containing $b_i^\tau$, the number of cliques will increase and the liveness will decrease.

However, our fitness computation model presented in \ref{alg:blockclique} states that, if a super-majority certificate $cert_i$ endorsing the block $b_i^\tau$ of the honest branch can be formed, even though this super-majority certificate has not yet been included in a block, it accounts for the fitness of the honest branch. Indeed, the block that he endorses is a leaf block. Thus correct endorsers, if they have received enough endorsements endorsing $b_i^\tau$ to form a super-majority certificate, will continue endorsing the honest branch rather than the attack branch. If so, the block producer has the incentive to build on the honest branch, as he will be able to double its reward by including two super-majority certificates rather than one.

Figure \ref{fig:fork-attack} presents how our incentive model fosters correct processes to adopt the honest branch of the fork rather than the attackers.

\begin{figure*}
  \centering
  \includegraphics[width=0.8\linewidth]{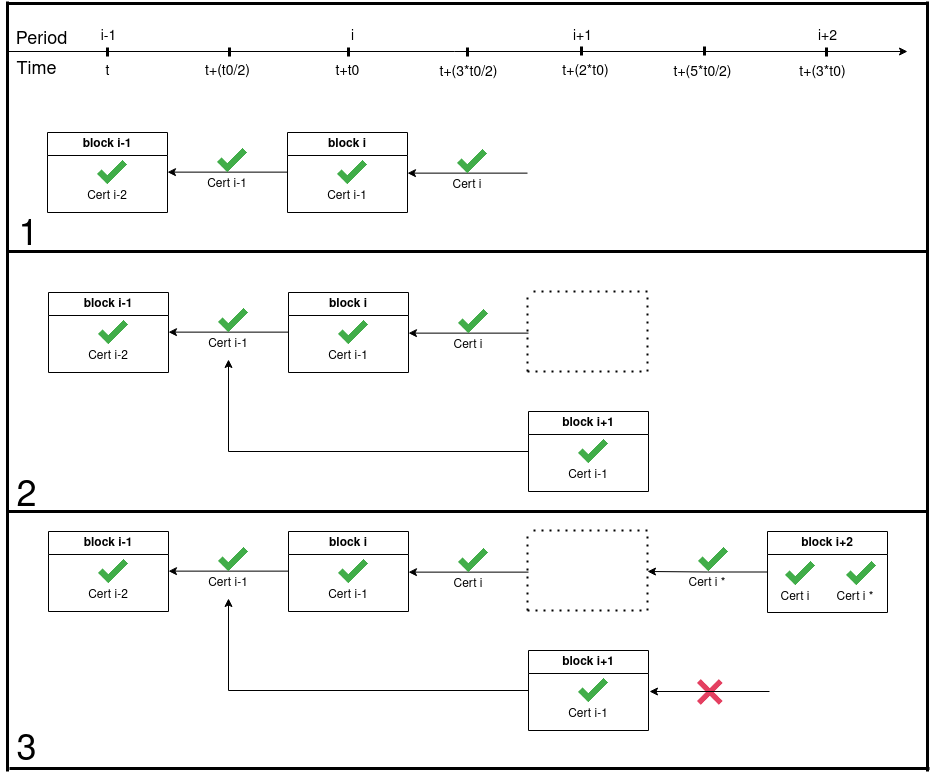}
  \caption{Fork attack
  \\
    \textbf{Frame 1:} Before the beginning of the attack, some block $b_i$ has been appended to the chain. It includes a super-majority certificate $cert_{i-1}$ that endorses his thread parent block $b_{i-1}$. A time $t+(2*t_0/2)$, the block $b_i$ has received enough endorsements to form a super-majority certificate. 
  \\
    \textbf{Frame 2:} The attackers launches its fork attack by broadcasting block $b_{i+1}$. Instead of choosing $b_i$ as a parent, which is the best parent according to protocol rules, the attacker chooses $b_{i-1}$ as a parent and includes $cert_{i-1}$ in its block $b_{i+1}$. This results in a fork with the honest branch containing $b_i$ and the attack branch $b_{i+1}$. Both blocks are incompatible because they have the same parent.
  \\
    \textbf{Frame 3:} From the point of view of an honest endorser, the honest branch of the fork has greater fitness. Indeed, $b_i$ received enough endorsement to form a super-majority certificate, and because $b_i$ is a leaf node, according to the fitness computation rules (See algorithm \ref{alg:blockclique}), $cert_i$ increases the fitness of the honest branch by one. Because of this fitness boost, the blockclique contains the honest branch. Correct and rational endorsers of period $i+1$ will endorse block $b_i$ rather than $b_{i+1}$. If so, a second super-majority certificate endorsing $b_i$ can be produced. If the block producer of period $i+2$ is correct, he will select $b_i$ as parent and include both certificates to double its reward. If the block producer for period $i+2$ is Byzantine, he cannot build upon the attack branch as it has not received a super-majority certificate. 
    }
  \label{fig:fork-attack}
\end{figure*}

\section{Determining the committee and the super-majority}\label{app:super-majority-size}

Let $E$ be the size of the committee (i.e. the number of endorsement slots per block) and $Q$ be the super-majority (i.e. the threshold number of endorsements required for a block to be considered valid). $Q$ is expressed as a portion of $E$ such that $0 \leq Q \leq 1$
We identified two main threats to be considered when choosing values for $Q$ and $E$. One threat to the safety and one to liveness.

\begin{itemize}
    \item \textbf{Safety:} The first threat is the event during which an attacker controls more than $Q*E$ endorsements slots for a block he produces. In this case, the attacker can produce and broadcast arbitrary many versions of a block for the same slot. This would open up various attacks and reorgs opportunities such as multi-stacking, double spending, or attacks on finality.
    
    \item \textbf{Liveness:} The second threat is the case of an attacker controlling more than $(1-Q)*E$ endorsement slots for a block that he does or does not produce. In this case, by not broadcasting the block or endorsements, the attacker could harm the liveness of the protocol.
\end{itemize}

In other words, the larger $E$ is, the safer the protocol, but the more CPU, storage, and network are required. The larger $Q$ is, the safer is the protocol, but the higher the probability of the attacker harming the liveness of the protocol. 

We can see that both those threats are negatively correlated; one is dependent on $Q$ while the other is on $1-Q$. As often in blockchain systems, there is a dilemma between safety and liveness. 

We model the attacker as an entity controlling a portion $\beta$ of the total stake, that always behaves in the optimal way to perform a given attack and that seeks to disturb the functioning of the protocol for its benefit or even at its own cost.

Our stance is that, against such a hypothetical attacker, the protocol should be safe with a very high probability and be live eventually. This choice is not supported by any theoretical proof and anybody could argue the contrary. But it is backed by the idea that in practice, the cost-to-benefit ratio of a liveness attack will discourage any rational attacker. 

Figure \ref{fig:N-T-safety} presents the safety of the protocol given various values of $Q$ and $E$ and given a worst-case attacker that controls a portion $\beta = 1/3$ of the total stake.

\begin{figure*}[!htb]
    \centering
    \includegraphics{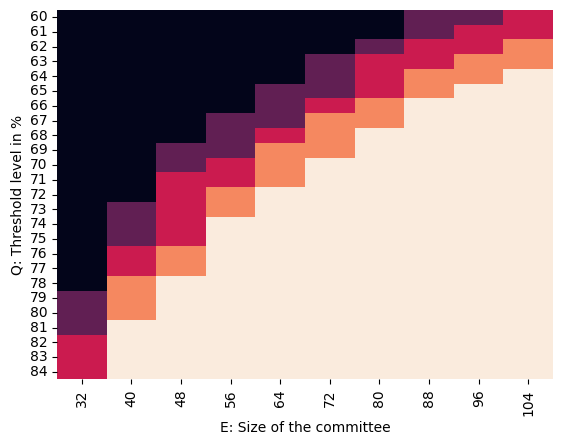}
    \caption{This graph represents the probability of an attacker that controls $1/3$ of the total stake to have more than $Q*E$ endorsement slots for a given block. Each color represents a different security level. In white, the attacker has more than $Q*E$ endorsements for a given slot on average once every $10^4$ years. In orange every $10^3$, etc.
 }
    \label{fig:N-T-safety}
\end{figure*}

\section{Corollaries}
\begin{lemma}[Limit on the number of multi-staked blocks]
As long as the attacker does not reach the threshold in the committee, the number of blocks created at each slot that passes the threshold of endorsements is finite.
\end{lemma}

\begin{proof}
Let $Q$ be the endorsement threshold. By hypothesis, the attacker can only produce $Q-k$ with $k > 0$ endorsements. Since honest nodes produce only a single endorsement per slot, it follows that the number of blocks that can reach the endorsement threshold is $\lfloor\frac{n+2k}{k}\rfloor$.
\end{proof}

\begin{lemma}
An attacker cannot construct an alternative chain with higher fitness than he later discloses.
\end{lemma}

\begin{proof}
It follows immediately from the parameters chosen for the size of the committee and of the super-majority: an attacker cannot reach the endorsement threshold by himself and thus cannot create an alternative chain that overtakes the blockclique.
\end{proof}
\section{Endorser and Block producer Algorithms}

\begin{algorithm*}[!htb]
\caption{Endorser instance for Slot: s and PoSIndex: index}
\label{alg:endorser}
\footnotesize
\begin{algorithmic}[1]
\state
\State $Endorsement$ : 2D vector of endorsements $Endorsements[slot][index]
\rightarrow Endorsement$ \Comment{2D vector because for each block slot, there are $E$ endorsement slots}
\State $endorsed$ : Vector of slots for which endorsement has already been broadcast
\\
\function{\textbf{main}$(s, index)$}
\If{$(s, index) \notin endorsed$} 
    \State $blockClique \gets blockClique()$ \Comment{Compute blockclique}
    \State $blockToEndorse \gets$ hash of last block $\in blockClique$ of thread $s.thread$ \Comment{Select best thread parent}
    \State $e \gets Endorsement(s, index, blockToEndorse)$ \Comment{Instantiate endorsement}
    \broadcast[ENDORSEMENT, $e$] \Comment{Broadcast endorsement}
    \State $endorsed.append((s, index))$
\EndIf
\end{algorithmic}
\end{algorithm*}

\begin{algorithm*}[!htb]
\caption{Block producer instance for Slot: s}
\label{alg:block-producer}
\footnotesize
\begin{algorithmic}[1]
\state
\State $Endorsements$ : 2D vector of endorsements $Endorsements[slot][index]
\rightarrow Endorsement$
\State $blockDenunciations$ : Vector of Byzantine blocks
\State $endorsementDenunciations$ : Vector of Byzantine endorsements
\\
\function{\textbf{main}$(s)$}
\State $blockClique \gets blockClique()$ \Comment{Compute blockclique}
\State $parents \gets$ each last block of each thread $\in blockClique$ that received a super-majority\Comment{\parbox[t]{.2\linewidth}{Select best parent}}
\State $operations \gets$ select thread operations from the pool (including denunciations)
\State $certs \gets [\;]$
\For{$endorsementSlot \in Endorsements$} \Comment{For each endorsement slot}
    \If{$endorsementSlot.thread = s.thread$} \Comment{If slot of current thread}
        \State $certificate \gets buildCert(endorsementSlot)$ \Comment{Try to build certificate}
        \If{$certificate \land certificate.endorsedBlock = parents[s.thread]$}
             \State $certs.append(certificate)$ \Comment{Append certificate}
        \EndIf
    \EndIf
\EndFor
\State $b \gets (s, parents, certs, operations)$ \Comment{Instantiate block}
\broadcast[BLOCK, $b$] \Comment{Broadcast block}
\end{algorithmic}
\end{algorithm*}

\section{Utilities full description}
Algorithm \ref{alg:utilities-full} implements the Interface presented in \ref{alg:utilities-interface}.
\begin{algorithm*}[!htb]
\caption{Utilities full implementation}
\label{alg:utilities-full}
\footnotesize
\begin{algorithmic}[1]
\state
\State $G$ : Block graph
\State $G_C$ : compatibility graph
\State $time$ : Number of milliseconds since the genesis block of thread 0
\State $pubKey$ : This process public key
\State $threshold$ : Threshold number of endorsements to form a super-majority certificate
\\
\function{$timeToSlot(time)$} \Comment{Transforms time into slot}
    \State \Return $slot$
    
\\
\function{$buildCert(endorsements)$} \Comment{Tries to build a super-majority certificate out of an endorsement list}
\State $cert \gets Certificate()$ \Comment{Instantiate certificate}
\State $hash \gets $ most duplicated key $endorsedBlock \in endorsements$
\State $cert.endorsement.append(\forall e \in endorsements$ where $e.endorsedBlock = hash$)
\If{$\lvert cert.endorsements \rvert \geq threshold$} \Comment{If threshold reached}
    \State $cert.endorsedBlock \gets hash$ \Comment{Set cert. parameters}
    \State $cert.slot \gets $ slot of the included endorsements
    \State \Return $cert$ \Comment{Return certificate}
\Else
    \State \Return false \Comment{Return false if certificate cannot be built}
\EndIf

\\
\function{$committeeBlocProducer(Slot: s)$} \Comment{PoS draw block producer}
    \State \Return $blockProducer$ \Comment{Returns the block producer public key}
    
\\
\function{$committeeEndorsers(Slot: s)$} \Comment{PoS draw endorsers given slot and ledger state}
    \State \Return $endorsers$ \Comment{Returns a vector of objects $(pubKey, index)$}
\\

\function{$parentsOlderThanBlock(Block: b)$}
\For{$parent \in b.parents$}
    \State \textbf{require} $b.slot > G[parent].slot$
\EndFor

\\
\function{$parentsMutuallyCompatible(Block: b)$} \Comment{Checks if a block's parents are all mutually compatible}
\For{$someParent \in b.parents$}\Comment{For each parent}
    \For{$someOtherParent \in b.parents$}\Comment{For each other parent}
        \State \textbf{require} $G_C.getEdge(someParent, someOtherParent)$ \Comment{Verify that they are mutually compatible}
    \EndFor
\EndFor
\State \Return $true$

\\
\function{$grandParentOlderThanParent(Block: b)$}
\For{$parent \in b.parents$} \Comment{For each parent}
    \For{$grandParent \in G[parent].parents$} \Comment{For each grand parent}
        \State \textbf{require} $parent.slot \geq grandParent.slot$ \Comment{Verify that parent is more or equally recent than grand parent}
    \EndFor
\EndFor
\State \Return $true$

\\
\function{$certificatesEndorseRightParent(Block: b)$}
\For{$cert \in b.certificates$} \Comment{For each included certificate}
    \State \textbf{require} $isValid(cert)$ \Comment{Verify that certificate is valid}
    \State \textbf{require} $cert.endorsedBlock = b.parents[b.slot.thread].hash$ \Comment{Verify that each certificate endorses thread parent}
\EndFor
\State \textbf{require} $\exists cert \in b.certificates$ where $cert.slot =  b.parents[b.slot.thread].slot$  \Comment{\parbox[t]{.3\linewidth}{Verify that at least one $cert$ was produced during same period as the block $b$ thread parent}}
\State \Return $true$
\end{algorithmic}
\end{algorithm*}

\end{appendices}
\end{document}